\documentclass[journal]{IEEEtran}
\ifCLASSINFOpdf
\else
\fi
\usepackage{amsmath,amssymb,amsfonts,amsthm}
\usepackage{graphicx}
\usepackage{textcomp}
\usepackage{xcolor}
\usepackage{comment}
\usepackage{bm,bbm}
\usepackage{wrapfig}
\usepackage{subfigure}
\usepackage{caption}
\usepackage[ruled,vlined]{algorithm2e}
\usepackage{makecell}
\usepackage{booktabs}       
\usepackage{balance}

\newtheorem{definition}{Definition}
\newtheorem{assumption}{Assumption}
\newtheorem{theorem}{Theorem}
\newtheorem{corollary}{Corollary}[theorem]
\newtheorem{lemma}{Lemma}

\DeclareMathOperator{\Lap}{Lap}

\graphicspath{{imgs/}}
\begin{document}

\title{On the Practicality of Differential Privacy in Federated Learning by Tuning  Iteration Times}
%
%
%
%

\author{Yao Fu,
        Yipeng Zhou,
        Di Wu,
        Shui Yu,
        Yonggang Wen,
        and Chao Li
\IEEEcompsocitemizethanks{
\IEEEcompsocthanksitem Yao Fu and Di Wu are with the Department of Computer Science, Sun Yat-sen University, Guangzhou, 510006, China, and Guangdong Key Laboratory of Big Data Analysis and Processing, Guangzhou, 510006, China (E-mail: fuyao7@mail2.sysu.edu.cn; wudi27@mail.sysu.edu.cn). 
\IEEEcompsocthanksitem Yipeng Zhou is with the Department
of Computing, FSE, Macquarie University, Australia, 2122 (E-mail: yipeng.zhou@mq.edu.au).
\IEEEcompsocthanksitem Shui Yu is with the School of Computer Science, University of Technology Sydney, Australia (E-mail: Shui.Yu@uts.edu.au).
\IEEEcompsocthanksitem Yonggang Wen is with the School of Computer Science and Engineering, Nanyang Technological University, Singapore (E-mail: ygwen@ntu.edu.sg).
\IEEEcompsocthanksitem Chao Li  is with Tencent Technology (Shenzhen) Co. Ltd, China.(E-mail:  ethancli@tencent.com).
}
}

%
%

\markboth{Journal of \LaTeX\ Class Files,~Vol.~14, No.~8, August~2015}%
{Shell \MakeLowercase{\textit{et al.}}: Bare Demo of IEEEtran.cls for Computer Society Journals}
%



\IEEEtitleabstractindextext{%
\begin{abstract}
In spite that Federated Learning (FL) is well known for its privacy protection when training machine learning models  among distributed clients collaboratively, recent studies have pointed out that the naive FL is susceptible to gradient leakage attacks. In the meanwhile,  Differential Privacy (DP) emerges as a promising countermeasure to defend against gradient leakage attacks. However, the adoption of DP by clients in FL may significantly jeopardize the model accuracy. It is still an open problem to understand the practicality of DP from a theoretic perspective. In this paper, we make the first attempt to understand the practicality of DP in FL through tuning the number of  conducted iterations. Based on the FedAvg algorithm,  we formally derive the convergence rate with DP noises in FL.  Then, we theoretically derive: 1) the conditions for the DP based FedAvg to converge as the number of global iterations (GI) approaches infinity; 2) the method to set the number of local iterations (LI) to minimize the negative influence of DP noises.   By further substituting the Laplace and Gaussian mechanisms into the derived convergence rate respectively, we show that: 3) The DP based FedAvg with the Laplace mechanism cannot converge, but the divergence rate can be effectively prohibited by setting the number of LIs with our method; 4) The learning error of the DP based FedAvg with the Gaussian mechanism can converge to a constant number finally if we use a fixed number of LIs per GI. 
To verify our theoretical findings, we conduct extensive experiments using two real-world datasets. The results not only validate our analysis results, but also provide useful guidelines on how to optimize model accuracy when incorporating DP into FL.
\end{abstract}

\begin{IEEEkeywords}
Federated Learning, Differential Privacy, Convergence Rate, Model Accuracy.
\end{IEEEkeywords}}

\maketitle

\IEEEdisplaynontitleabstractindextext

%
\IEEEpeerreviewmaketitle

\section{Introduction}\label{sec:introduction}

%
%
%
%
\IEEEPARstart{I}{n}  the past decade, machine learning models have demonstrated unprecedented data processing capability in numerous applications. To well train machine learning models, it is unavoidable to massively collect data samples from users, which gives rise to the concern on user privacy leakage \cite{yeom2018privacy,hitaj2017deep}. This concern has significantly hindered the wide application of machine learning techniques. In light of this, 
\emph{Federated Learning (FL)} \cite{pmlr-v54-mcmahan17a} has been proposed to allow decentralized clients to collaboratively train machine learning models by merely exchanging intermediate computations (\emph{i.e.}, gradients) with the parameter server (PS). Raw data samples are  locally kept on clients. However, it has been pointed out that the leaked gradient information can  still be exploited by malicious entities to crack user privacy \cite{wei2020framework, zhao2020idlg, zhu2019deep,236216,8737416}.\footnote{Users are used interchangeably with clients in this paper.} To further enhance the protection of user privacy,  \emph{Differential Privacy (DP)} has been employed to add additional noises to client gradients before they are uploaded to the PS \cite{bhowmick2018protection,wu2019value,wei2019federated}. 


Unfortunately, the adoption of DP  may lower the model accuracy significantly, which has become a major obstacle for the application of DP in FL~\cite{bhowmick2018protection,wu2019value,wei2019federated}. Taking the popular FedAvg algorithm~\cite{pmlr-v54-mcmahan17a} as an example, the naive adoption of DP in FedAvg may cause $2$-$10$ times training loss than the original model~\cite{wu2019value}. For many real-world applications, such high performance deterioration can significantly reduce the effectiveness of advanced machine learning models, and is somewhat unacceptable at all. As FedAvg represents a large family of gradient descent (GD) based algorithms widely used in FL, it is questionable how practical DP will be when being applied to FL.

Typically, clients and the PS in FL works together with the following manner: clients conduct a number of local iterations (LI) before their gradients plus DP noises are uploaded to the PS for aggregation. By receiving the computations from multiple selected clients, the PS conducts a round of  aggregation and distributes new results to clients for the next round of global iteration (GI) \cite{kairouz2019advances}. 
Apparently, it is only necessary to add the DP noises  to clients' gradients after the last round of LIs before they are disclosed to the PS. The noises distort the disclosed gradients against attacks, which meanwhile impair the convergence of FL \cite{wu2019value,wei2019federated,geyer2017differentially}.\footnote{Note that clients may have different privacy budgets and are independent in noise generation \cite{7113353,10.1145/2676726.2677005}} 

In this work, we investigate how to improve the practicality of DP in FL through tuning the number of conducted local or global iterations. Intuitively, we can improve the model accuracy if the number of LIs is tuned properly since there is no need to add noises after each LI except the last LI. 
Firstly, we derive the convergence rate of the FedAvg algorithm distorted by the DP noises. Secondly, we derive the conditions for the DP based FedAvg to converge, and formally derive the formula to set the number of LIs according to the DP mechanism to minimize the negative influence of DP noises. At last, 
we conduct case study by using  the Laplace mechanism and the Gaussian mechanism, through which we theoretically unveil that: 1) The learning error of the DP based FedAvg with the Gaussian mechanism can finally converge to a constant number if we use a fixed  number of LIs per GI; 2) The DP based FedAvg with the Laplace mechanism will diverge with the number of GIs, but the divergence rate can be lowered substantially by setting the number of LIs according to our method.


Overall, our contributions in this paper can be summarized as below:
\begin{itemize}
\item To the best of our knowledge, we are the first to theoretically study the practicality of DP in FL through tuning the numbers of LIs and GIs, which paves the way towards a better understanding of the complicated relation between the random noises introduced by DP and the final FL model accuracy. 
\item We formally derive the convergence rate of FedAvg with DP noises and the convergence conditions. By substituting the Laplace and Gaussian mechanisms into the convergence rate, we theoretically investigate how to optimally tune the number of LIs and GIs so as to minimize the influence of DP noises.
\item Extensive experiments with real-world  datasets are conducted to validate the correctness of our theoretic results, and demonstrate the practical merit of our work, which  provides us a way to optimize the FL model accuracy even if DP is adopted. 
\end{itemize}


The rest content is organized as follows. The related state-of-the-art works are discussed in Sec. \ref{Sec:related}. The preliminary knowledge is introduced in Sec. \ref{sec:preliminary}.  The DP-FedAvg algorithm is elaborated in Sec. \ref{sec:algorithm}.  The derivation and optimization of the converge rate of DP-FedAvg  are presented in Sec. \ref{sec:fedavg}. The experiment results are presented and explained in Sec. \ref{sec:experiment} before we finally conclude our paper in Sec. \ref{sec:conclusion}. 

 

\section{Related Work}\label{Sec:related}
\subsection{Federated Learning}

The FL framework was originally proposed in the works \cite{pmlr-v54-mcmahan17a,googleFL, bonawitz2019towards}, which applied FL on mobile devices to predict the next word of mobile users' text input. Both FedAvg and FedSGD algorithms used for neural network based models were empirically studied with extensive experiments in~\cite{pmlr-v54-mcmahan17a}. Later on, the convergence rate of FedAvg/FedSGD was analyzed in \cite{Li2020On} with non-IID data sample distributions.

Since the inception of FL, it has attracted tremendous attentions. \emph{Kairouz et al.}~\cite{kairouz2019advances} and \emph{Li et al.}~\cite{Li2020FederatedLC}  conducted a holistic overview of FL with  in-depth technical discussions on possible weaknesses of FL. It has been proved that FL can improve the protection of user privacy to some extent \cite{kairouz2019advances}. However, the leakage of the gradient information can still result in the leakage of user privacy, which can be exploited by malicious users to design various attack strategies \cite{wei2020framework, zhao2020idlg, zhu2019deep,236216,8737416}.

\subsection{Differential Privacy}
Other than FL, DP is another widely adopted approach to preserve user privacy \cite{10.1007/11787006_1}. The principle of DP is to add noises to individual users' private information before it is disclosed. Consequently, malicious attackers cannot exactly crack users' private information even if the information is leaked. Thanks to its ability for preserving privacy, it has been applied in personalized recommendation systems \cite{8290673}, location based services \cite{ding2017collecting}, meta learning \cite{li2020differentially} and so on. However, the shortcoming of DP is that the service quality will be lowered  because service providers cannot obtain precise user information as well \cite{bhowmick2018protection, wu2019value,NEURIPS2019_fc0de4e0}.

\subsection{Applying DP in FL}

DP can be applied to machine learning such as FL to enhance the privacy protection level. In FL, there exist two schemes to add DP noises: the global DP scheme and the local DP scheme. The former one adds noises to the aggregated gradient information on the server side before the gradient  information is distributed to clients \cite{geyer2017differentially,brendan2018learning}. Whereas the latter one adds noises to the gradient information on users  before it is uploaded to the server \cite{bhowmick2018protection,pihur2018differentiallyprivate,10.1145/3378679.3394533,seif2020wireless}. 

There are many existing  works \cite{Abadi_2016,geyer2017differentially,brendan2018learning} that targeted to design mechanisms to add  global DP noises in machine learning systems such as FL.  For example, the work \cite{geyer2017differentially} particularly focused on using DP to mask whether a client participates in the training (\emph{i.e.}, client-level DP).

Prior works on the incorporation of LDP are elaborated as follows. \emph{Wu et al.} \cite{wu2019value} theoretically analyzed the convergence performance when local DP is used in a distributed learning system.  Only the gradient descent algorithm is considered in this work, and its analysis proved that the GD algorithm can converge if the sample population approaches infinity, which however is impossible in practice. In addition, the scenario considered in this paper is not a typical FL setting, and the influence of the non-IID sample distribution is ignored. 
\emph{Wei et al.} \cite{wei2019federated} considered both global DP and local DP in an FL system. Their theoretical analysis validated that the number of global iterations and the number of engaged clients in each global iteration should be neither too large nor too small to minimize the final loss function. However, only the FedSGD algorithm with  IID sample distribution was mentioned in this work, which was not consistent with an FL scenario.
The analysis of the FedAvg algorithm and the influence of the number of local iterations were absent in the above work. 



Other than the above LDP mechanism that can guarantee the privacy budget over all rounds of global iteration, there exist  a number of other works that designed and analyzed a loose LDP mechanism that can only  guarantee the privacy budget for a single round of global iteration~\cite{bhowmick2018protection,seif2020wireless,liu2020fedsel}.  For example, \emph{Seif et al.} \cite{seif2020wireless} relaxed the local $\epsilon$-DP in wireless FL settings, which gave fixed noises in each iteration and used advanced composition rule to track the total privacy leakage over $T$ iterations. However,  the privacy leakage will increase to infinity with $T$, though they can obtained a converged upper bound finally. 

In summary,  prior works have demonstrated that it is promising to incorporate DP into FL. However, how to minimize the negative influence of the noises incurred by DP has been overlooked by existing works. Especially, how to set up the number of LIs and how it affects the final model accuracy have not been explored yet, and this gap will be filled by our work. 

Before we introduce our algorithm, we list the symbols and notations used in this paper, as shown in Table~\ref{tab:notations}. Some ofthem will be defined later in the following sections.

\begin{table}[!htbp]
\centering
\caption{Notations}
\label{tab:notations}
\begin{tabular}{cm{0.8\linewidth}}
\toprule
Symbol & Description\\
\midrule
$\mathcal{N}$ & the set of all clients \\
$N$ & the number of all clients \\
$\mathfrak{M}(\theta)$ & the model with parameters $\theta$\\
$p$ & the dimension of the parameter\\
$l$ & the client\\
$n_l$ & the number of data owned by client $l$\\
$\mathcal{D}_l$ & the dataset of client $l$\\
$d$ & a batch of data\\
$f$ & the loss function\\
$\nabla f$ & the gradient of function $f$\\
$t$ & the index of GI\\
$k$ & the index of all iterations\\
$\mathcal{P}_t$ & the engagement pool in $t$-th GI\\
$b$ & the size of the engagement pool\\
$\eta$ & the learning rate\\
$T$ & the total number of iterations\\
$E$ &  the  number  of  LIs  to  be  executed \\
$T_g$ & the  total  number  of  GIs  to  be  executed\\
$T_l$ & the number of GIs that client $l$ will participate \\
$\epsilon,\delta$ & DP parameters\\
$\mathcal{M}$ & DP mechanism\\
$\mathcal{C}_E$ & the set of global communication round\\
$\mathbf{w}$ & the noise generated by DP mechanism\\
$\Gamma$ & the non-IID measurement\\
\bottomrule
\end{tabular}
\end{table}

\section{Preliminaries}\label{sec:preliminary}
\subsection{FedAvg}
In FL, there is a group of clients $\mathcal{N}$ with cardinality $N$ that aim to collaboratively train a model $\mathfrak{M}(\theta)$, where $\theta\in\mathbb{R}^p$ represents the parameters of the model. Each client $l$ has a private data set $\mathcal{D}_l$ with cardinality $n_l$. A trusted parameter server is responsible for sending model parameters to clients and aggregating updates reported back from clients in each round of global iteration. 
To train the model $\mathfrak{M}(\theta)$, which can be represented by the loss function $f(\theta)$, a common approach is to design gradient descent (GD) based algorithms (\emph{e.g.}, FedAvg, FedSGD \cite{pmlr-v54-mcmahan17a})  to iteratively reduce the loss function.

In FedAvg, each client updates the model on the local dataset $\mathcal{D}_l$ for $E$ epochs, following the rule:
\begin{equation*}
\theta\leftarrow\theta-\eta\nabla f(\theta;d),   
\end{equation*}
where $d\in \mathcal{D}_l$ represents a batch of data samples owned by client $l$ and $\eta$ is the learning rate. 

After every $E$ local iterations,  the PS conducts a global iteration by aggregating clients' models following the rule 
\begin{equation*}
\theta_{t+1}\leftarrow\frac{1}{\sum_{l\in\mathcal{P}_t}n_l}\sum_{l\in\mathcal{P}_t}n_{l}\theta_{t+1}^{l},
\end{equation*}
where $\mathcal{P}_t$ represents the set of clients engaged in the global iteration $t$. FedSGD is a special case of FedAvg by setting $E=1$ and $d = \mathcal{D}_l$. 
For both FedAvg and FedSGD, the PS can only communicate with a certain number of  clients, which may be randomly selected due to the limitation of communication capacity. 


To ease our discussion, we define the following notations. Let $T_g$ denote the total number of GIs to be executed,   $E$  denote the number of LIs to be executed by each engaged client per GI and $b$ denote the number of clients engaged by the PS in each round of GI. 
Note that $b$ clients execute $E$ local iterations in parallel, and thus the total number of iterations $T$ to be executed is equal to the product of $E$ and $T_g$, \emph{i.e.}, $T = ET_g$.


\subsection{Differential Privacy}

Both FedAvg and FedSGD assume that the PS can be trusted and clients' privacy cannot be cracked through gradient information. However, such assumptions are not always valid \cite{zhu2019deep, zhao2020idlg,236216,8737416}. Thereby, the DP technique is introduced to add noises to gradients generated by clients before their gradients are aggregated. 

We briefly introduce the essential concepts in  DP~\cite{10.1007/11787006_1,dwork2014algorithmic} before we introduce how to incorporate DP into FedAvg. 
\begin{definition}[($\epsilon,\delta$)-Differential Privacy]
A randomized function $\mathfrak{Q}$ gives $\epsilon$-differential privacy if for any data sets $\mathcal{D}_1$ and $\mathcal{D}_2$ differing on at most one entry, and all $\mathcal{S}\subseteq Range(\mathfrak{Q})$,
\begin{equation*}
    \Pr\left\{\mathfrak{Q}(\mathcal{D}_1)\in \mathcal{S}\right\}\le\exp(\epsilon)\times\Pr\left\{\mathfrak{Q}(\mathcal{D}_2)\in \mathcal{S}\right\}+\delta.
\end{equation*}
\end{definition}

The traditional Laplace mechanism gives $(\epsilon, 0)$-DP, while the Gaussian mechanism gives $(\epsilon,\delta)$-DP. By tolerating some possibility of fails (represented by $\delta$), the Gaussian mechanism performs better on the utility under composition. We will introduce detailed Laplace and Gaussian mechanisms later. For generality, we use $\mathcal{M}$ to denote the noise generator of DP mechanisms.




\section{Client Based DP-FedAvg Algorithm }
\label{sec:algorithm}

In this section, we  introduce the client based DP-FedAvg algorithm and the implementation of two most frequently used DP mechanisms, \emph{i.e.}, Laplace and Gaussian. 

\subsection{ Algorithm Framework}

Before the algorithm details are presented, we define some useful notations as follows.
The index set of FL GIs is defined as $\mathcal{C}_E=\left\{1*E,2*E,\dots,T_g*E\right\}$. For any iteration in $\mathcal{C}_E$, the PS aggregates parameters returned by clients.  For the sake of clarity, we use two indicators $t$ and $k$ throughout this paper where $k=0, 1, \dots, T$ represents the index of any LI and $t=0,1,2,\dots,T_g$ represents  the index of any GI. The relation between $k$ and $t$ is $t=\lfloor \frac{k}{E}\rfloor$. It implies that the $k$-th LI occurs just before the $t$-th  GI where $t=\lfloor \frac{k}{E}\rfloor$. In each round of GI, we formally define the set of clients who are engaged to exchange gradients with the PS as the engaged client pool (ECP) denoted as $\mathcal{P}_t$ with cardinality $b$.
In our work, the clients in  
$\mathcal{P}_t$ are selected from the complete client set $\mathcal{N}$ with a round robin manner. 

The LIs are executed by a client only if the client is engaged in $\mathcal{P}_t$. The client receives the latest model from the PS (denoted by the parameter $\theta_{tE}$ where $tE$ represents the total number of LIs conducted until the $t$-th GI). Then, the client updates the model on the local dataset for $E$ rounds before the client returns the updated model with noises added with the rule  $\hat{\theta}_{(t+1)E}^l=\theta_{(t+1)E}^l+\mathbf{w}_{t}^l$. Here $\mathbf{w}_{t}^l$ represents the DP noises which will be further introduced later. 
The PS updates the model with the same rule of the FedAvg algorithm by aggregating models returned from engaged clients. 
The  details of the client based DP-FedAvg algorithm are presented in Algorithm \ref{alg:fedminibatchsgd}.

\begin{algorithm}
\SetAlgoLined
\DontPrintSemicolon

\SetKwProg{Fn}{Server executes}{:}{}
\Fn{}{
 initialize $\theta_0$\;
\For{$t\leftarrow 0$ \KwTo\ $T_g-1$}{
 choose $b$ clients as a group $\mathcal{P}_t\subset\mathcal{N}$ with a round robin manner\;
 \ForEach{client $l\in\mathcal{P}_t$ \textbf{in parallel}}{
  $\hat{\theta}_{(t+1)E}^l,n_l\leftarrow$ClientUpdate($\theta_{tE},\,E$)\;
 }
 $\theta_{(t+1)E}\leftarrow\frac{N}{b}\sum_{l\in\mathcal{P}_t}\frac{n_l}{n}\hat{\theta}_{(t+1)E}^l$\;
 }
}
\;
\SetKwProg{Fn}{ClientUpdate}{:}{}
\Fn{$(\theta_k,\,E)$}{
$\theta_k^l\leftarrow\theta_k$\;
\For{$i=k$ \KwTo $k+E-1$}{
$\theta_{i+1}^l\leftarrow\theta_{i}^l-\eta_i\nabla f_l(\theta_i^l)$
}
\For{$i\leftarrow0$ \KwTo $p-1$}{
$\mathbf{w}_{t}^l[i]\leftarrow \mathcal{M}$\;
}
$\hat{\theta}_{k+E}^l\leftarrow \theta_{k+E}^l+\mathbf{w}_{t}^l$\;
\KwRet{$\hat{\theta}_{k+E}^l$}
}

\caption{The Client Based DP-FedAvg Algorithm. $f_l(\theta_i^l)$ is a simplified form of $f(\theta_i^l;\mathcal{D}_l)$.}
\label{alg:fedminibatchsgd}
\end{algorithm}

To facilitate our discussion, we present the overall iteration rules of Algorithm \ref{alg:fedminibatchsgd} here again as
\begin{equation}\label{equ:update}
\begin{split}
&\nu_{k+1}^l=\theta_{k}^l-\eta_{k}\nabla f_l(\theta_{k}^l),\\
&\theta_{k+1}^l=\begin{cases}
\nu_{k+1}^l, &\text{if }k+1\notin\mathcal{C}_E, \\
\frac{N}{b}\sum_{l\in\mathcal{P}_t}\frac{n_l}{n}\nu_{k+1}^l+\mathbf{w}_{t}^b, &\text{if }k+1\in\mathcal{C}_E,
\end{cases}
\end{split}
\end{equation}
where
\begin{equation*}
 \mathbf{w}_{t}^b=\frac{N}{b}\sum_{l\in\mathcal{P}_t}\frac{n_l}{n}\mathbf{w}_{t}^l,
\end{equation*}
is the average of noises from $b$ engaged clients. Here, for convenience, we use $k$ to denote the iteration index regardless of global or local iterations. For $k+1 \in \mathcal{C}_E$, local parameters represented by  $\theta_{k+1}^l$ will be aggregated with other engaged clients and DP noises. 

Different from the original FedAvg algorithm, the DP noise (\emph{i.e.}, $\mathbf{w}_{t}^b=\frac{N}{b}\sum_{l\in\mathcal{P}_t}\frac{n_l}{n}\mathbf{w}_{t}^l$) is the factor that can impede the convergence of DP-FedAvg. We name the  term $\mathbf{w}_{t}^b$ as  the ``\emph{noise item}".



\subsection{DP Mechanisms} 

The DP mechanism determines how to set $\mathbf{w}_{t}^l$ on each client before the updated model is returned to the server. In this work, we consider two most frequently used mechanisms (\emph{i.e.}, Laplace and Gaussian) in the convergence analysis. In fact, other DP mechanisms can be substituted into our framework as well  by setting $\mathbf{w}_{t}^b$ with a different form.   

According to the composition rule of DP \cite{dwork2014algorithmic}, we cannot consider the DP privacy budget for each iteration separately. 
Suppose that the PS conducts $T_g$ rounds of GIs, then each client will participate $T_l$ rounds of GIs where  $T_l=\frac{bT_g}{N}$. Here we assume that the PS selects $b$ clients for every GI in a round robin manner and $N$ is the total number of clients in the system. 
Let $\tilde{\eta}_t=\max_{tE\le k<(t+1)E}(\eta_k)$ where $\eta_k$ denote the learning rate of the  $k$-th iteration. 
The Laplace and the Gaussian mechanisms can be defined as below. 
\begin{theorem}[Laplace Mechanism]\label{theorem:Lap}
Let $\xi_1$ denote the $L1$-sensitivity of $\nabla f_l$. Algorithm~\ref{alg:fedminibatchsgd} is $\epsilon$-differential private over $T_l$ rounds if $\mathcal{M}=\Lap\left(0, \frac{T_l}{\epsilon}\Xi_1\right)$ 
where $\Xi_1=\tilde{\eta}_tE\xi_1$ is the $L1$-sensitivity of $\theta^l_{tE} - \theta^l_{(t+1)E}$ after the client conducts $E$ rounds of LIs.
\end{theorem}
\begin{corollary}\label{cor:noise_lap}
The variance of the noise item of the Laplace Mechanism is
\begin{equation*}
    \mathbb{E}\left\{\left\|\mathbf{w}_t^b\right\|_2^2\right\}=2pb\frac{\Xi_1^2T_g^2}{n^2\epsilon^2}\bar{n}^2,
\end{equation*}
where $p$ is the dimension of $\theta$, $n = \sum_{l=1}^N n_l$ and $\bar{n}^2 = \frac{1}{N}\sum_{l=1}^N n_l^2$.
\end{corollary}
The detailed proof of Theorem~\ref{theorem:Lap} is presented in Appendix~\ref{proof:theorem_Lap}.

\begin{theorem}[Gaussian Mechanism~\cite{Abadi_2016}]\label{theorem:Ga}
Assume the gradient is bounded as $\|\nabla f_l\|_2\le \xi_2$.
There exist constants $c_1$ and $c_2$ so that given the sampling probability $q$ on each client and the number of GIs $T_l$ that client $l$ participates, for any $\epsilon<c_1q^2T_l$, Algorithm~\ref{alg:fedminibatchsgd} is $(\epsilon,\delta)$-differential private for any $\delta>0$ if $\mathcal{M}=\mathcal{N}(0,\sigma^2\Xi_2^2)$ with
\begin{equation*}
    \sigma= c_2\frac{q\sqrt{T_l\log(1/\delta)}}{\epsilon},
\end{equation*}
where $\Xi_2=\tilde{\eta}_tE\xi_2$.   
\end{theorem}
In Algorithm~\ref{alg:fedminibatchsgd}, each client uses the full batch of samples for LIs implying that the sampling probability $q=1$.
\begin{corollary}\label{cor:noise_gau}
The variance of the noise item of the Gaussian Mechanism is
\begin{equation*}
    \mathbb{E}\left\{\left\|\mathbf{w}_t^b\right\|_2^2\right\} = 2pc_2^2N\log(1/\delta)\frac{\Xi_2^2T_g}{n^2\epsilon^2}\bar{n}^2.
\end{equation*}
\end{corollary}
The detailed proof of Theorem~\ref{theorem:Ga} is presented in Appendix~\ref{proof:theorem_Ga}.


To simplify our analysis, we define the asymptotic variance of the DP noise item as follows.
\begin{definition}
\label{DEF:AsyVar}
The asymptotic variance of the noise item of a DP mechanism $\mathcal{M}$ is defined as $\mathbb{V}_{\mathcal{M}} =  O(\tilde{\eta}_t^2E^2T_g^z)$, where $z \in [0, 2]$ is a  constant number determined by the DP mechanism.  
\end{definition}
In the definition of the asymptotic variance, all variables not related with the number of LIs or GIs are regarded as constant numbers and ignored. 
We define $\mathbb{V}_{\mathcal{M}}$ based on the fact that $\tilde{\eta}_t^2E^2$ affects the function sensitivity, while $T_l= \frac{b}{N}T_g$ represents the number of GIs each client needs to participate. Meanwhile, we let $0\leq z\leq 2$ because  the privacy budget at most increases linearly with $T_l$ according to the composition rule.\footnote{According to the composition rule in the DP theory, $\epsilon$-DP will be at most $T_l\epsilon$-DP if the database is queried for $T_l$ times. }

This is a very generic definition that can depict most existing DP mechanisms devised for FL, such as \cite{wu2019value,wei2019federated}. The asymptotic variances of the noise items of the Laplace and Gaussian mechanisms introduced as above can be  obtained as
\begin{equation}
\label{EQ:AVLapalce}
    \mathbb{V}_{\textit{Laplace}}=O(\tilde{\eta}_t^2E^2T_g^2).
\end{equation}
\begin{equation}
\label{EQ:AVGaussian}
    \mathbb{V}_{\textit{Gaussian}}=O(\tilde{\eta}_t^2E^2T_g) .
\end{equation}

Obviously, how to set $E$ and $T_g$ will affect the variance of the noise item, and hence it will affect the convergence of DP-FedAvg. Meanwhile, we note that the asymptotic variance of the Gaussian mechanism is lower than that of the Laplace mechanism,  and therefore the Gaussian mechanism can yield a better convergence rate.
We conduct the convergence analysis of DP-FedAvg in the next section.


\section{Convergence Analysis of Client Based DP-FedAvg}\label{sec:fedavg}

In this section, we give the theoretic analysis of the convergence of FedAvg when DP mechanisms are applied on the client side.
\subsection{Definitions and Assumptions}
In FL, it is well-known that the distribution of data samples is non-IID, that is, the training data samples of different clients may be drawn from distinct distributions. By leveraging the  measure of the degree of non-IID in~\cite{Li2020On}, we define the degree of non-IID as below. 
\begin{definition}\label{def:non-iid} The degree of non-IID is quantified by
\begin{equation*}
    \Gamma=f^*-\sum_{l\in\mathcal{N}}\frac{n_l}{n}f_l^*,
\end{equation*}
where $f^*$ represents the loss function with global optimal parameters and $f_l^*$ represents the loss function with local optimal parameters.
\end{definition}

To ease our discussion, we define $\bar{\theta}_k=\sum_{l\in\mathcal{N}}\frac{n_l}{n}\theta_k^l$ and $\bar{\nu}_k=\sum_{l\in\mathcal{N}}\frac{n_l}{n}\nu_k^l$ to denote the global parameters by averaging over all clients. Similarly, we also define
$\bar{\nu}_k^b=\frac{N}{b}\sum_{l\in\mathcal{P}_t}\frac{n_l}{n}\nu_k^l$ to denote the sampled global parameters. We use $\nabla f(\theta_k)=\sum_{l\in\mathcal{N}}\frac{n_l}{n}\nabla f_l(\theta_k^l)$ to denote the global gradient and $\nabla f_b(\theta_k)=\frac{N}{b}\sum_{l\in\mathcal{P}_t}\frac{n_l}{n}\nabla f_l(\theta_k^l)$ to denote the sampled global gradient. 
\begin{definition}
Define $Y_k=\mathbb{E}\left\{\left\|\bar{\theta}_k-\theta^*\right\|_2^2\right\}$ as the expected distance between the current, global parameters and the optimal parameters after $k$ iterations.  
\end{definition}

Similar to the assumptions made in many previous papers~\cite{JMLR:v19:17-650, pmlr-v80-nguyen18c,Nguyen2019TightDI,Li2020On}, we also make the same following assumptions to simplify our analysis.
\begin{assumption}\label{assumption:smooth}
The loss function $f$ is $\lambda$-smooth. Formally, there exists a constant $\lambda>0$, and
\begin{equation*}
    f(\theta)\le f(\theta')+\left<\nabla f(\theta'),\theta-\theta'\right>+\frac{\lambda}{2}\left\|\theta-\theta'\right\|_2^2.
\end{equation*}
\end{assumption}
\begin{assumption}\label{assumption:strong_convex}
The loss function $f$ is $\mu$-strongly convex. Formally, there exists a constant $\mu>0$, and 
\begin{equation*}
    f(\theta)\ge f(\theta')+\left<\nabla f(\theta'),\theta-\theta'\right>+\frac{\mu}{2}\left\|\theta-\theta'\right\|_2^2.
\end{equation*}
\end{assumption}


\begin{assumption}\label{assumption:bound of gradient}
	The expectation of gradients is bounded in terms of
	\begin{equation*}
	\mathbb{E}\left\{\left\|\nabla f_l(\theta_k^l)\right\|_2^2\right\}\le G^2,\,\forall l\in\mathcal{N},0\le k<T.
	\end{equation*}
\end{assumption}

\subsection{Convergence Rate of DP-FedAvg}

To derive the upper bound of the convergence rate  of Client based DP-FedAvg, we need to utilize the following lemmas. 
\begin{lemma}[Unbiased sample]\label{lemma:unbaised sample}
	The round robin sampling of clients is unbiased, and thus we have
	\begin{equation*}
	\mathbb{E}\left\{\bar{\nu}_k^b\right\}=\bar{\nu}_k.
	\end{equation*}
\end{lemma}
\begin{lemma}[Bounding the divergence of  the parameters of sampled clients in $\mathcal{P}_t$]\label{lemma:sample variance}
Let Assumption~\ref{assumption:bound of gradient} hold, the gap between $\bar{\nu}_{k+1}^b$ and $\bar{\nu}_{k+1}$  is bounded as
\begin{equation*}
    \mathbb{E}\left\{\left\|\bar{\nu}_{k+1}^b-\bar{\nu}_{k+1}\right\|_2^2\right\}\le4E^2\eta_{k}^2G^2\frac{N-b}{N-1}\frac{1}{b}.
\end{equation*}
\end{lemma}
\begin{lemma}[Bounding the divergence of local parameters]\label{lemma:variance of E}
Let Assumption~\ref{assumption:bound of gradient} hold, 
the expectation of the average distance between $\bar{\theta}_k$ and $\theta_{k}^l$'s  is 
\begin{equation*}
\mathbb{E}\left\{\sum_{l\in\mathcal{N}}\frac{n_l}{n}\left\|\bar{\theta}_k-\theta_k^l\right\|_2^2\right\}\le 4\eta_{k}^2(E-1)^2G^2.
\end{equation*}
\end{lemma}
\begin{lemma}[Upper bound in one step of FedAvg]\label{lemma:upper bound of one step}
Let Assumption~\ref{assumption:smooth},~\ref{assumption:strong_convex} and~\ref{assumption:bound of gradient} hold, in one step of DP-FedAvg, we have
\begin{equation*}
\begin{split}
\left\|\bar{\nu}_{k+1}-\theta^*\right\|_2^2&=(1-\mu\eta_{k})\left\|\bar{\theta}_k-\theta^*\right\|_2^2\\
&\quad+2\sum_{l\in\mathcal{N}}\frac{n_l}{n}\left\|\bar{\theta}_k-\theta_k^l\right\|_2^2+6\lambda\eta_{k}^2\Gamma.
\end{split}
\end{equation*}
\end{lemma}
The detailed proof of Lemma~\ref{lemma:unbaised sample} is presented in Appendix. The proof of Lemma~\ref{lemma:sample variance} and~\ref{lemma:variance of E} can be found in~\cite{Li2020On}.  Lemma~\ref{lemma:upper bound of one step} can be proved by modifying  the proof of in~\cite{Li2020On} a little bit. The only difference is that each client executes local iterations with a full batch of samples owned by the client in our case. 
\begin{theorem}\label{theorem:upper bound}
Let Assumptions~\ref{assumption:smooth},~\ref{assumption:strong_convex} and~\ref{assumption:bound of gradient} hold, and $\eta_k=\frac{2}{\mu}\frac{1}{k+\gamma}$ where $\gamma=\max(8\frac{\lambda}{\mu},E)$. The convergence rate of the DP-FedAvg algorithm is bounded by 
\begin{equation*}
\begin{split}
&Y_k\le\frac{1}{k+\gamma}\left(\frac{4}{\mu^2}\omega_0+\gamma Y_0\right)+\frac{4}{\mu^2}\frac{t}{(k+\gamma-1)^2}\omega_1,\\
&\omega_0=6\lambda\Gamma+8(E-1)^2G^2+4E^2G^2\frac{N-b}{N-1}\frac{1}{b},\\
\end{split}
\end{equation*}
where $\omega_1=C_\mathcal{M}E^2T_g^z $ and $C_\mathcal{M}$ represents a constant number related with the DP mechanism. For example, $C_\mathcal{M}=8pb\frac{\xi_1^2}{N^2\epsilon^2}$ for Laplace mechanism and $C_\mathcal{M}=8pc_2^2\log(1/\delta)\frac{\xi_2^2}{N\epsilon^2}$ for Gaussian mechanism. Here  $z$ is a constant number defined in Definition~\ref{DEF:AsyVar}, and its value is dependent on the DP mechanism. 
\end{theorem}
Here, we derive the convergence rate of DP-FedAvg with a general DP mechanism. $t$ is the index of the global iteration, and thus $t= \lfloor \frac{k}{E}\rfloor$.
As we can see from the convergence rate, $\omega_0$ governs the convergence of DP-FedAvg without the noise item; while $\omega_1$ is the influence of the noise item on the convergence rate. 
The proof of Theorem~\ref{theorem:upper bound} leverages Lemmas~\ref{lemma:unbaised sample}, \ref{lemma:sample variance}, \ref{lemma:variance of E} and \ref{lemma:upper bound of one step}, which can be found in the Appendix with details.

\subsection{Asymptotic Analysis}
To focus on the influence of iteration times, we conduct asymptotic analysis of the convergence rate by regarding all variables not related to the LI or GI as constant numbers.  
\subsubsection{Convergence Conditions}
The convergence rate derived in Theorem~\ref{theorem:upper bound} can be simplified as
\begin{eqnarray}
\label{EQ:AsympYT}
    Y_T &=& O\left(\frac{E^2}{T}\right ) +O\left(\frac{\omega_1}{ET}\right ),\nonumber\\
    & = &  O\left(\frac{E^2}{T}\right ) +O\left(E^{1-z}T^{z-1}\right).
\end{eqnarray}
Here, we utilize the fact that $\omega_1 = O(E^2T_g^z)$ and  $ET_g= T$ is the total number of iterations. Through Eq.~\eqref{EQ:AsympYT}, we prove the convergence conditions of DP-FedAvg as follows.
\begin{theorem}
\label{Them:ConvConditons}
$Y_T$ of the client based DP-FedAvg with the mechanism $\mathcal{M}$ will converge to $0$ as $T$ approaches infinity, if $z<1$ in $\mathbb{V}_{\mathcal{M}}$. The number of LIs should be set as  $E= O(T^{\frac{z}{z+1}} )$ to achieve the fastest convergence rate  $O(T^{\frac{z-1}{z+1}})$.
\end{theorem}
\begin{proof}

We first prove the convergence conditions. 
From the first term in Eq.~\ref{EQ:AsympYT}, we have $E<O(\sqrt{T})$, otherwise the first term will diverge.  Thus, for the second term we have $O(E^{1-z}T^{z-1} ) < O(T^{\frac{1-z}{2}+z-1} )$. It converges only if $\frac{1-z}{2}+z-1<0$ which implies that $z<1$. 

Then, we prove the optimal value of $E$ is $O(T^{\frac{z}{z+1}} )$. 
Since $Y_T = O(\frac{E^2}{T})+O(E^{1-z}T^{z-1})$, it is a convex function with respect to $E$. By letting the  differentiation of $Y_T$ with respect to $E$ equal to $0$, we have $O(\frac{E}{T})+O(E^{-z}T^{z-1}) = 0$. It turns out that we should set $E= O(T^{\frac{z}{z+1}} )$. 
\end{proof}

\begin{theorem}
\label{The:DivSetE}
$Y_T$ of the client based DP-FedAvg with the mechanism $\mathcal{M}$ cannot converge to $0$ as $T$ approaches infinity, if $z\geq 1$ in $\mathbb{V}_{\mathcal{M}}$. The number of LIs should be set as  $E= O(T^{\frac{z}{z+1}} )$ to minimize the influence of the DP noise item, which gives rise to the lowest divergence  rate  $O(T^{\frac{z-1}{z+1}})$.
\end{theorem}
\begin{proof}
From the proof of Theorem~\ref{Them:ConvConditons}, we can see that $Y_T$ cannot converge if $z\geq 1$.

Similar to the proof of Theorem~\ref{Them:ConvConditons}, we should set $E= O(T^{\frac{z}{z+1}} )$ to achieve $Y_T=O(T^{\frac{z-1}{z+1}})$. However, due to $z\geq 1$,  $Y_T$ will diverge with the minimum rate 
$O(T^{\frac{z-1}{z+1}})$ as $T$ approaches infinity. 
\end{proof}

\begin{corollary}
If $z>1$ in $\mathbb{V}_{\mathcal{M}}$, there exists a $T^*<\infty$ that can minimize $Y_T$.
\end{corollary}
The proof is straightforward given that $Y_T$ is a convex function with respect to $T$ and $Y_T$ will diverge if $z>1$. {Unfortunately, it is extremely difficult to precisely derive $T^*$ due to the difficulty to exactly estimate the value of each parameter such as $\mu$, $\Gamma$ and $G$ in $Y_T$.} 


\subsubsection{Case Study}

For the Laplace and the Gaussian mechanisms, they are the special cases of our general analysis.  
For the Laplace mechanism, we have $z=2$, which implies that the Laplace mechanism can never converge. The divergence rate can be minimized by setting $E$ as an increasing function of $T$. According to the Theorem~\ref{The:DivSetE}, we should set $E=O(T^{\frac{2}{3}})$ and $Y_T$ will diverge with the lowest rate $O(T^{\frac{1}{3}})$ as $T$ approaches infinity.

\noindent{\bf Remark:} We can see the significance of our study from This case study. If we set $E = O(1)$ for the Laplace mechanism, $Y_T$ will diverge with the rate $O(T)$, which is much higher than the optimal rate $O(T^{\frac{1}{3}})$. Thus, setting $E$ properly with our method can markedly  improve the practicability of DP in FL.  

The DP based FedAvg with the Gaussian mechanism is a very special case.  According to the Theorem~\ref{The:DivSetE}, $Y_T$ cannot converge to $0$. However, if we set $E= O(1)$, $Y_T$ will converge to $O(1)$ as $T$ approaches infinity.

\section{Experiment}\label{sec:experiment}

In this section, we carry out extensive experiments with two real-world datasets to validate our theoretical analysis. 

\subsection{Dataset Description and Preprocessing}

\subsubsection{Lending Club Dataset}

The Lending Club dataset~\cite{wu2019value} contains about $890,000$ loan data records. Each  record in the dataset includes some features of borrowers (\emph{e.g.}, zip code, address, request amount of money, and interest rate).

The dataset is used to predict customers' interest rates. We first clean the data by discarding meaningless features (\emph{e.g.}, ID, address, and job title). Then we convert some valuable text features to numbers, such as job year and grade. We discard those features missed by more than $70\%$ of data records. After the above preprocessing operations, there remain $66$ features, though some features are redundant and some are almost always $0$. We use Principal Component Analysis (PCA) to further reduce $66$ features into $10$ effective features. In the experiments, we use randomly sampled $500,000$ data records as the whole dataset. And then we randomly sample $80\%$ data as the training dataset and the rest as the test dataset. To setup the non-IID distribution, all records in the training dataset are sorted  by the interest rate in ascending order before they are evenly grouped into  $N=5,000$ groups. Each client will be assigned a single sample group. 

\subsubsection{ FEMNIST Dataset}
FEMNIST is an open source dataset presented in~\cite{caldas2018leaf} as a benchmark for FL study. FEMNIST contains $803,267$ hand-written image samples written by $3,500$ users. There are $62$ labels (\emph{i.e.}, $10$ digits, $26$ lowercase, $26$ uppercase) in the dataset. We use $90\%$ samples to train the model and the rest as the test set. Because FEMNIST is naturally unbalanced and non-IID, we simulate $3,500$ clients and assign each client with a user's samples.

\subsection{Models and Experimental Settings}

We construct a Linear Regression (LR) model with $11$ parameters for the Lending Club dataset to predict the interest rate of each borrower.  Mean Square Error (MSE) is used to define the loss function and evaluate the model prediction performance. 

For the FEMNIST dataset, we construct a Convolutional Neural Network (CNN) model to classify images. The CNN model consists of two convolutional layers followed by one fully connected layer. The first convolution layer uses $7\times7$ kernels with padding $3$. The second layer uses $3\times3$ kernels with padding $1$. The first layers has $32$ channels while the second has $64$ channels. Both convolutional layers are followed by a ReLU and $2\times2$ max pool. The fully connected layer has $62$ units and takes the flatted output of the second convolutional layer as the input. The CNN model contains $214,590$ parameters.
CrossEntropy is used to define the loss function.



In the experiments, we use the MSE loss function  and the model accuracy on the test datasets to evaluate the LR model and CNN model, respectively. We can evaluate the model performance over the test set after each round of global iteration.
If $E>1$, the model is only evaluated when $k\in\mathcal{C}_E$.

For each case, we implement a benchmark model, which is a special case with $\epsilon = \infty$, \emph{i.e}., there  is no DP noise. 
To get rid of the influence of the randomness due to random noises and random client selection in ECPs,  we plot the average results by executing each experiment case for $20$ times. 
By default, we set $b=500$ for the LR model and $b=350$ for the CNN model. 

We impose the Laplace mechanism on the LR model and the Gaussian mechanism on the CNN model. Throughout the following experiments, we set a fixed $\delta=0.0001$ for the Gaussian mechanism.

In practice, it is very difficult to exactly compute the sensitivity, \emph{i.e.}, $\xi$, of the gradient functions, and thus we use the norm clipping technique with a clipping threshold $\zeta$ to restrict the range of client's gradients. If a client's some gradient exceeds $\zeta$, it will be clipped to $\zeta$. 
Formally, the gradient is clipped by 

\begin{equation*}
 \nabla {f}_l(\theta_k)=\nabla f_l(\theta_k)/\max\left(1,\frac{\left\|\nabla f_l(\theta_k)\right\|_1}{\zeta}\right),   
\end{equation*}

where we set $\zeta=150$ for LR model and $\zeta=3$ for CNN model. 
This technique is firstly introduced by \cite{Abadi_2016} and widely used in the implementation of DP in practical machine learning systems \cite{brendan2018learning,geyer2017differentially}.

\subsection{Results and Discussions}
In this part, we present the results of our experiments.

\subsubsection{Impact of $\epsilon$}
To inspect the influence of the privacy budge on the model training process, we conduct experiments by enumerating different $\epsilon$'s. For the LR model, we set $T=100, E=1$ and $\epsilon=1.0, 3.0, 5.0$ or $\infty$; while for the CNN model, we set $T=250, E=5$ and $\epsilon=10.0, 20.0, 30.0$ or $\infty$. 
Note that we set relatively large $\epsilon$'s for the CNN model due to the large number of parameters in CNN. Similar $\epsilon$'s were also used in other works such as~\cite{wei2019federated,wu2019value}.  

The loss function of the LR model is obtained on the test set of  the Lending Club dataset. In Fig.~\ref{fig:Lap_exp1}, we can observe that there exists a gap between the loss function curves with DP noises and the loss function curve with $\epsilon = \infty$. In particular, the loss function diverges rapidly when $\epsilon$ is very small. It shows the challenge to make DP based FL practical, especially when the privacy budget is small. 

\begin{figure}
    \centering
    \includegraphics[width=\linewidth]{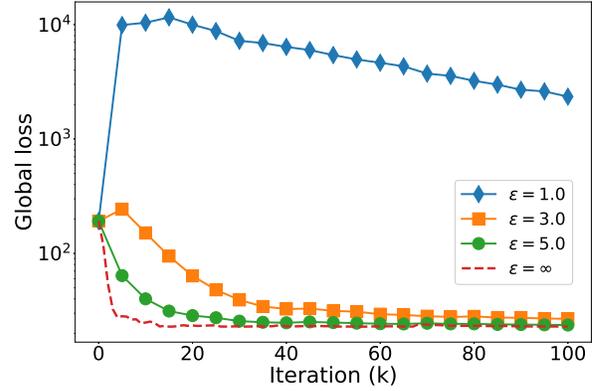}
    \caption{Comparing the loss function of  DP-FedAvg in each iteration $k$ with $T=100$ and different $\epsilon$'s using the Lending Club dataset and Laplace mechanism.. $\epsilon=\infty$ represents the noise-free FedAvg.  }
    \label{fig:Lap_exp1}
\end{figure}

The accuracy  of the CNN model evaluated on the FEMNIST test set  is plotted in Fig.~\ref{fig:Gau_exp1} against the iteration times. In Fig.~\ref{fig:Gau_exp1}, we can draw a similar conclusion that the model accuracy can be impaired by the DP noises significantly such that  implementing DP in FL is challenge when $\epsilon $ is small. 
\begin{figure}
    \centering
    \includegraphics[width=\linewidth]{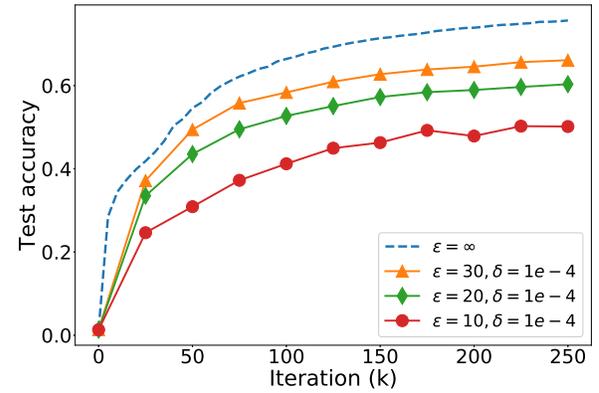}
    \caption{Comparing the model accuracy of DP-FedAvg in each iteration $k$ with $T=50$ and different $\epsilon$'s using the FEMNIST dataset and the Gaussian mechanism. $\epsilon=\infty$ represents the noise-free FedAvg.}
    \label{fig:Gau_exp1}
\end{figure}

\subsubsection{Impact of $T$}

We conduct experiments to verify Theorem~\ref{The:DivSetE}, which states that: 1) $Y_T$ approaches infinity as $T$ increases and  there exists some $T^*<\infty$ such that $Y_T$ is optimized if the Laplace mechanism is adopted; 2) $Y_T$ converges to $O(1)$ as $T$ approaches infinity if the Gaussian mechanism is adopted.

\begin{figure}
    \centering
    \includegraphics[width=\linewidth]{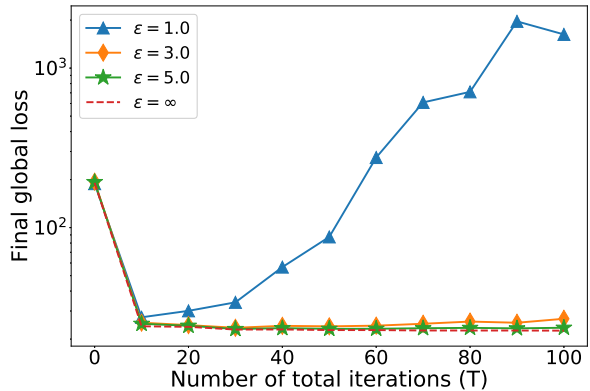}
    \caption{Comparing the loss function of DP-FedAvg after $T$ iterations with different $\epsilon$'s using the Lending Club dataset and the Laplace mechanism.}
    \label{fig:Lap_exp2}
\end{figure}

\begin{table}[htbp]
    \centering
        \caption{Final global loss of the experiment in Fig.~\ref{fig:Lap_exp2}. }
    \label{tab:Lap_exp2}
    \begin{tabular}{llll}
    \toprule
       $\epsilon$ & $T^*$ & final global loss of $T^*$  & \bf{worst} \\
    \midrule
        $1.0$ & 10  & 27.40&1963.41 \\
        $3.0$ & 30 & 23.55 &26.82 \\
        $5.0$ & 30 & 23.13 &24.94 \\
        $\infty$ & 100  & 22.59 &24.06 \\
    \bottomrule
    \end{tabular}
\end{table}

In Fig.~\ref{fig:Lap_exp2}, we evaluate the LR model using the Lending Club dataset and the Laplace mechanism by varying $T$ from $0$ to $100$ with a step size $10$ and set $E=1$. We plot the loss function $Y_T$ against $T$.  {Note that $Y_T$ is different from $Y_k$ where $Y_T$ is the final loss after $T$ iterations and $Y_k$ is the loss after  $k$ iterations. } We also show the accurate value of $T^*$ and the corresponding loss in Table~\ref{tab:Lap_exp2}.
As we can see that the loss function diverges finally as $T$ increases. For each experiment case, there exists some $T^*$ such that the minimum value of the loss function is achieved. 

For the  CNN model, we use the FEMNIST dataset for evaluation by varying $T$ from $0$ to $500$ with a step size of $50$ and $E=5$. As shown in Fig.~\ref{fig:Gau_exp2}, the model accuracy gets stable at some value as $T$ approaches infinity  per the Gaussian mechanism is adopted. This result  implies that $Y_T$ approaches $O(1)$ instead of diverging to infinity as $T$ approaches infinity.

\begin{figure}
    \centering
    \includegraphics[width=\linewidth]{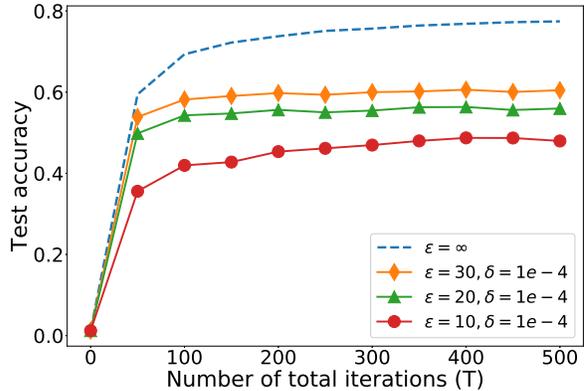}
    \caption{Comparing the model accuracy of DP-FedAvg after $T$ iterations with different $\epsilon$'s by using the FEMNIST dataset and the Gaussian mechanism.}
    \label{fig:Gau_exp2}
\end{figure}

\subsubsection{Impact of $E$}
We conduct experiments to show the importance to tune $E$, \emph{i.e.}, the number of local iterations, in DP-FedAvg.  

We  enumerate different $E$'s to see how $E$ can affect the model training process.
In Fig.~\ref{fig:Lap_exp3}, we evaluate  the LR model using the Lending Club dataset and the Laplace mechanism, by enumerating $E=1,2,3,4,6,12$ and  fixing $T=120$ with different $\epsilon$'s. 
Similarly, in Fig.~\ref{fig:Gau_exp3}, we evaluate the  CNN model by fixing $T=240$ and enumerating $E=1,2,3,4,6,8,12,24$ with different $\epsilon$'s. To accurately show the optimal $E$, we list the results in Fig.~\ref{fig:Lap_exp3} and Fig.~\ref{fig:Gau_exp3} in Table~\ref{tab:Lap_exp3} and Table~\ref{tab:Gau_exp3}.

Through observing Fig.~\ref{fig:Lap_exp3} and Fig.~\ref{fig:Gau_exp3}, we can draw the following observations.
\begin{itemize}
    \item In both Fig.~\ref{fig:Lap_exp3} and Fig.~\ref{fig:Gau_exp3},  for all cases with different $\epsilon$'s,  if $E$ has  not been set properly, the performance of the trained model can be very poor. For example, the loss function of the LR model is very high if $E=1$ when $\epsilon =1.0$.  This experiment manifests the importance to tune $E$ in DP based FL.  
    
    \item $E$ should be neither too large or too small. Intuitively, if $E$ is too large, it means $T_g$ will be too small and the model training process is not iterated sufficiently. Inversely, if $E$ is too small, it means $T_g$ is too large, and hence the model training process can be impaired  by the large DP noises. Recall that $\mathbb{V}_{\mathcal{M}} =O(\tilde{\eta}_t^2E^2T_g^z)$.
    
    \item Through comparing cases with different $\epsilon$'s in both figures, we can find the rule of thumb that $E$ should be set a larger value if $\epsilon $ is smaller. 
    
\end{itemize}

\begin{figure}
    \centering
    \includegraphics[width=\linewidth]{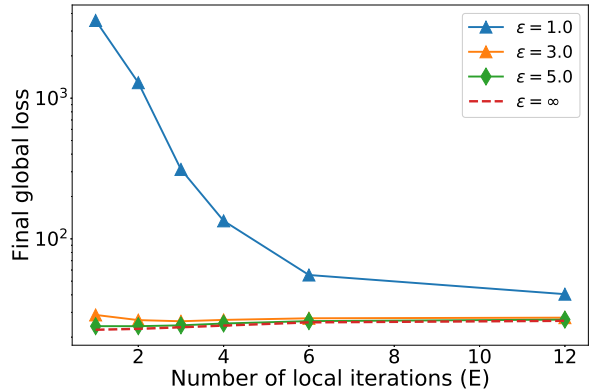}
    \caption{Comparing the final global loss of DP-FedAvg after $T$ iterations with various $E$'s, fixed $T=120$ and different $\epsilon$'s by using the Lending Club dataset and the Laplace mechanism.}
    \label{fig:Lap_exp3}
\end{figure}

\begin{table}[htbp]
    \centering
        \caption{Final global loss of the experiment in Fig.~\ref{fig:Lap_exp3}. }
    \label{tab:Lap_exp3}
    \begin{tabular}{llll}
    \toprule
       $\epsilon$ & $E^*$ & final global loss of $E^*$  & \bf{worst} \\
    \midrule
        $1.0$ & 12  & 40.44&3568.41 \\
        $3.0$ & 3 & 25.95 &28.78 \\
        $5.0$ & 2 & 23.94 &26.65 \\
        $\infty$ & 1  & 22.58 &26.14 \\
    \bottomrule
    \end{tabular}
\end{table}

\begin{figure}
    \centering
    \includegraphics[width=\linewidth]{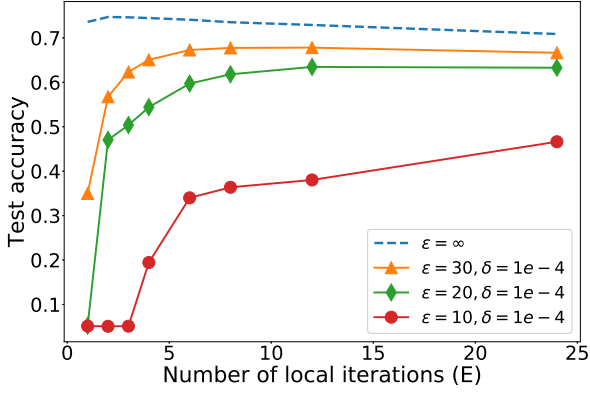}
    \caption{Comparing the model accuracy of DP-FedAvg after $T$ iterations with various $E$'s, fixed  $T=50$ and different $\epsilon$'s by using FEMNIST and the Gaussian mechanism.}
    \label{fig:Gau_exp3}
\end{figure}

\begin{table}[htbp]
    \centering
        \caption{Test accuracy of the experiment in Fig.~\ref{fig:Gau_exp3}. }
    \label{tab:Gau_exp3}
    \begin{tabular}{llll}
    \toprule
       $\epsilon$ & $E^*$ & test accuracy of $E^*$  & \bf{worst} \\
    \midrule
        $10$ & 24  & 0.4663&0.0509 \\
        $20$ & 12 & 0.6347 &0.0521 \\
        $30$ & 12 & 0.6781 &0.3500 \\
        $\infty$ & 2  & 0.7470 &0.7088 \\
    \bottomrule
    \end{tabular}
\end{table}

To verify Theorem~\ref{Them:ConvConditons}, we conduct the experiment by using the Lending Club dataset and the Laplace mechanism. In this experiment, we set $E$ as a function of $T$. Specifically, we enumerate $E = 1, T^\frac{1}{3}, T^\frac{1}{2}, T^\frac{2}{3} \textit{ or } T$ by varying  $T$ from $1$ to $120$ with a step size of $10$. Note in this experiment, we set $b=5000$ to enable very small $T_g$ (\emph{i.e.}, $T_g=1$). If $E$ (as a function of $T$) is not an integer, it will be rounded to the nearest integer. 
According to Theorem~\ref{Them:ConvConditons}, it is optimal to set $E^*=T^{2/3}$. The result is presented in Fig.~\ref{fig:Lap_exp4}, which indeed confirms the correctness of our theoretical analysis. The curve with $E=T^{\frac{2}{3}}$ is the best one achieving the lowest loss function. Although, $Y_T$ will diverge  with the increase of $T$, its divergence rate is very low if $E=T^{\frac{2}{3}}$. Another important implication is that the setting of $T$ is not very sensitive if $E$ is chosen properly. 

\begin{figure}
    \centering
    \includegraphics[width=\linewidth]{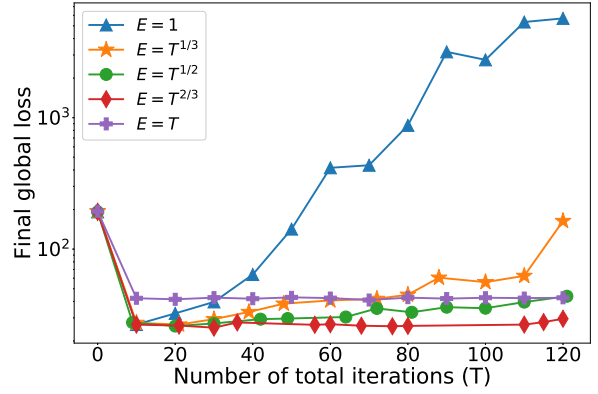}
    \caption{Comparing the final global loss of DP-FedAvg after $T$ iterations with a fixed $\epsilon=3$ and different $E$'s by using the Lending Club dataset and the Laplace mechanism. }
    \label{fig:Lap_exp4}
\end{figure}

\section{Conclusion}\label{sec:conclusion}

We investigate the practicability of DP mechanisms in FL. A general DP based FedAvg algorithm is presented first. Then, both the Laplace mechanism and the Gaussian mechanism are introduced into the general DP based FedAvg algorithm.
Through tuning the number of local/global iterations, we show from both theoretical analysis and experimental study that it is very important to set the number of iteration times properly such that the influence of the DP noises can be minimized. Specifically, via defining the asymptotic variance of the DP noises, we derive the conditions for DP-FedAvg to converge, and the optimal choice of the number of local iterations. The Lending Club and FEMNIST datasets are utilized to conduct extensive experiments. The experiment results not only verify our theoretical analysis, but also demonstrates that the choice of the number of local iterations can heavily affect the model training performance.


%

\appendices
\section{Proof of Theorem~\ref{theorem:Lap}}\label{proof:theorem_Lap}
\begin{proof}
We first show that the Laplace mechanism gives $\frac{\epsilon}{T_l}$-DP in each global iteration. Assume $\mathcal{D}_l, \mathcal{D}_l'$ are two data set differing at most one data. We have
\begin{equation*}
\begin{split}
    &\left\|\mathfrak{Q}_{t}\left(\mathcal{D}_l\right)-\mathfrak{Q}_{t}\left(\mathcal{D}_l'\right)\right\|_1\\
    &=\left\|\theta_{(t+1)E}^l-\theta_{(t+1)E}^{l'}\right\|_1\\
    &=\left\|\sum_{i=tE}^{(t+1)E-1}\eta_i\left(\nabla f_{l'}\left(\theta_i^{l'}\right)-\nabla f_l\left(\theta_i^l\right)\right)\right\|_1\\
    &\le\sum_{i=tE}^{(t+1)E-1}\eta_i\left\|\nabla f_{l'}\left(\theta_i^{l'}\right)-\nabla f_l\left(\theta_i^l\right)\right\|_1\\
    &\le \sum_{i=tE}^{(t+1)E-1}\eta_i\xi_1\\
    &\le E\xi_1\tilde{\eta}_{t}=\Xi_1.\\
\end{split}
\end{equation*}
$\forall t,\, \forall \mathbf{z}\in\mathbb{R}^p$, we have
\begin{equation*}
\begin{split}
    &\frac{\Pr_{\mathcal{D}_l}(\mathbf{z})}{\Pr_{\mathcal{D}_l'}(\mathbf{z})}\\
    &=\prod_{i=0}^{p-1}\left(\frac{\exp\left(-\frac{|\mathfrak{Q}_{t}\left(\mathcal{D}_l\right)[i]-\mathbf{z}[i]|}{\beta_t^l}\right)}{\exp\left(-\frac{|\mathfrak{Q}_{t}\left(\mathcal{D}_l'\right)[i]-\mathbf{z}[i]|}{\beta_t^l}\right)}\right)\\
    &=\prod_{i=0}^{p-1}\exp\left(\frac{|\mathfrak{Q}_{t}\left(\mathcal{D}_l\right)[i]-\mathbf{z}[i]|-|\mathfrak{Q}_{t}\left(\mathcal{D}_l'\right)[i]-\mathbf{z}[i]|}{\beta_t^l}\right)\\
    &\le \prod_{i=0}^{p-1}\exp\left(\frac{\left|\mathfrak{Q}_{t}\left(\mathcal{D}_l\right)[i]-\mathfrak{Q}_{t}\left(\mathcal{D}_l'\right)[i]\right|}{\beta_t^l}\right)\\
    &=\exp\left(\frac{\epsilon\left\|\mathfrak{Q}_{t}\left(\mathcal{D}_l\right)-\mathfrak{Q}_{t}\left(\mathcal{D}_l'\right)\right\|_1}{\Xi_1 T_l}\right)\le \exp\left(\frac{\epsilon}{T_l}\right).
\end{split}
\end{equation*}
Corollary 3.15 in \cite{dwork2014algorithmic} proves this Laplace mechanism gives $\epsilon$-DP over $T_l$ rounds.
\end{proof}
\section{Proof of Theorem~\ref{theorem:Ga}}\label{proof:theorem_Ga}
\begin{proof}
We show that the accumulated gradient is bounded.
\begin{equation*}
\begin{split}
        \left\|\theta_{tE}^l-\theta_{(t+1)E}^l\right\|_2&=\left\|\sum_{i=tE}^{(t+1)E-1}\eta_i\nabla f_l\left(\theta_i^l\right)\right\|_2\\
        &\le\sum_{i=tE}^{(t+1)E-1}\eta_i\left\|\nabla f_l\left(\theta_i^l\right)\right\|_2\\
        &\le\sum_{i=tE}^{(t+1)E-1}\eta_i\xi_2\\
        &\le E\xi_2\tilde{\eta}_t=\Xi_2.
\end{split}
\end{equation*}
Then by utilizing the Theorem~1 in~\cite{Abadi_2016}, we show that the Gaussian mechanism gives $(\epsilon,\delta)$-differential privacy for each client.
\end{proof}
\section{Proof of Corollary~\ref{cor:noise_lap}}\label{proof:noise_lap}
\begin{proof}
We first show that for each client $l$,
\begin{equation*}
\begin{split}
    \mathbb{E}\left\{\left\|\mathbf{w}_{t}^l\right\|_2^2\right\}
    &=\mathbb{E}\left\{\sum_{i=0}^{p-1}\left|\mathbf{w}_{t}^l[i]\right|^2\right\}\\    
    &=p\mathbb{E}\left\{\left|\mathbf{w}_{t}^l[0]\right|^2\right\}\\
    &=2p(\beta_t^l\Xi_1)^2=2p\frac{\Xi_1^2T_l^2}{\epsilon^2},
\end{split}   
\end{equation*}
where we use the fact that $\mathbf{w}_{t+1}^l[i]$ are IID random variables. For $b$ clients, we have
\begin{equation*}
\begin{split}
    \mathbb{E}\left\{\left\|\mathbf{w}_{t}^b\right\|_2^2\right\}
    &=\frac{N^2}{b^2n^2}\mathbb{E}\left\{\left\|\sum_{l\in\mathcal{P}_t}n_l\mathbf{w}_{t}^l\right\|_2^2\right\}\\
    &=\frac{N^2}{b^2n^2}\sum_{l\in\mathcal{P}_t}\mathbb{E}\left\{\left\|n_l\mathbf{w}_{t}^l\right\|_2^2\right\}\\
    &=2pb\frac{\Xi_1^2T_g^2}{n^2\epsilon^2}\bar{n}^2.
\end{split} 
\end{equation*}
The second equality holds because $\mathbf{w}_{t}^l$ are independent and $\mathbb{E}\left\{\mathbf{w}_{t}^l\right\}=0$.
\end{proof}

\section{Proof of Corollary~\ref{cor:noise_gau}}\label{proof:noise_gau}
\begin{proof}
We first show that for each client $l$,
\begin{equation*}
\begin{split}
    \mathbb{E}\left\{\left\|\mathbf{w}_{t}^l\right\|_2^2\right\}
    &=\mathbb{E}\left\{\sum_{i=0}^{p-1}\left|\mathbf{w}_{t}^l[i]\right|^2\right\}\\    
    &=p\mathbb{E}\left\{\left|\mathbf{w}_{t}^l[0]\right|^2\right\}\\
    &=2p(\sigma_t^l\Xi_2)^2\\
    &=2pc_2^2\log(1/\delta)\frac{\Xi_2^2T_l}{\epsilon^2},
\end{split}   
\end{equation*}
where we use the fact that $\mathbf{w}_{t+1}^l[i]$ are IID random variables. For $b$ clients, we have
\begin{equation*}
\begin{split}
    \mathbb{E}\left\{\left\|\mathbf{w}_{t}^b\right\|_2^2\right\}
    &=\frac{N^2}{b^2n^2}\mathbb{E}\left\{\left\|\sum_{l\in\mathcal{P}_t}n_l\mathbf{w}_{t}^l\right\|_2^2\right\}\\
    &=\frac{N^2}{b^2n^2}\sum_{l\in\mathcal{P}_t}\mathbb{E}\left\{\left\|n_l\mathbf{w}_{t}^l\right\|_2^2\right\}\\
    &=2pc_2^2N\log(1/\delta)\frac{\Xi_2^2T_g}{n^2\epsilon^2}\bar{n}^2.
\end{split} 
\end{equation*}
The second equality holds because $\mathbf{w}_{t}^l$ are independent and $\mathbb{E}\left\{\mathbf{w}_{t}^l\right\}=0$.
\end{proof}

\section{Proof of Lemma~\ref{lemma:unbaised sample}}
\begin{proof}
Taking expectation, we obtain
\begin{equation}\label{equ:unbiased 1}
\begin{split}
\mathbb{E}\left\{\bar{\nu}_k^b\right\}&=\frac{N}{b}\mathbb{E}\left\{\sum_{l\in\mathcal{P}_t}\frac{n_l}{n}\nu_k^l\right\}\\
&=\frac{N}{b}\sum_{\mathcal{P}\in\phi_b(\mathcal{N})}\frac{1}{c}\left(\sum_{l\in\mathcal{P}}\frac{n_l}{n}\nu_k^l\right),
\end{split}
\end{equation}
where $\phi_b(\mathcal{N})$ represents the set of all subset of $\mathcal{N}$ whose size is $b$, and $c=|\phi_b(\mathcal{N})|$. For instance, if $\mathcal{N}=\left\{1,2,3\right\}$, $\phi_2(\mathcal{N})=\left\{\left\{1,2\right\},\left\{1,3\right\},\left\{2,3\right\}\right\},c=3.$ Because we sample clients with equal probability and without replacement, each client appears identically and at most once in every $\mathcal{P}$. Therefore, since there are totally $cb$ clients sampled, each client is sampled $\frac{cb}{N}$ times. We change the order of summation and obtain 
\begin{equation*}
        \sum_{\mathcal{P}\in\phi_b(\mathcal{N}) }\sum_{l\in\mathcal{P}}\frac{n_l}{n}\nu_k^l=\sum_{l\in\mathcal{N}}\frac{cb}{N}\frac{n_l}{n}\nu_k^l.
\end{equation*}
Plugging this into (\ref{equ:unbiased 1}) and recalling the definition of $\bar{\nu}_k$, we finish the proof.
\end{proof}
\section{Proof of Theorem~\ref{theorem:upper bound}}\label{proof:theorem upper bound}
\begin{proof}
Based on the updating rule in (\ref{equ:update}), we have
\begin{equation}\label{equ:t3 update}
\begin{split}
\left\|\bar{\theta}_{k+1}-\theta^*\right\|_2^2=\begin{cases}
\left\|\bar{\nu}_{k+1}-\theta^*\right\|_2^2,&\text{if }k+1\notin\mathcal{C}_E,\\
\left\|\bar{\nu}_{k+1}^b+\mathbf{w}_{t}^b-\theta^*\right\|_2^2,&\text{if }k+1\in\mathcal{C}_E.\\
\end{cases}
\end{split}
\end{equation}

If $k+1\in\mathcal{C}_E$, we have
\begin{equation}\label{equ:t3 if in}
\begin{split}
\left\|\bar{\nu}_{k+1}^b+\mathbf{w}_{t}^b-\theta^*\right\|_2^2&=\left\|\bar{\nu}_{k+1}^b-\theta^*\right\|_2^2+\left\|\mathbf{w}_{t}^b\right\|_2^2\\
&\quad+2\left<\bar{\nu}_{k+1}^b-\theta^*,\mathbf{w}_{t}^b\right>.
\end{split}
\end{equation}
After taking expectation, the last term will vanish and the second term has been bounded by Corollary~\ref{cor:noise_lap} and~\ref{cor:noise_gau}. For the first term, we show that
\begin{equation*}
\begin{split}
\left\|\bar{\nu}_{k+1}^b-\theta^*\right\|_2^2&=\left\|\bar{\nu}_{k+1}^b-\bar{\nu}_{k+1}+\bar{\nu}_{k+1}-\theta^*\right\|_2^2\\
&=\left\|\bar{\nu}_{k+1}^b-\bar{\nu}_{k+1}\right\|_2^2+\left\|\bar{\nu}_{k+1}-\theta^*\right\|_2^2\\
&\quad+2\left<\bar{\nu}_{k+1}^b-\bar{\nu}_{k+1},\bar{\nu}_{k+1}-\theta^*\right>.
\end{split}
\end{equation*}

Taking expectation on both sides, using Lemma~\ref{lemma:unbaised sample} and combining Lemma~\ref{lemma:sample variance},~\ref{lemma:variance of E} and~\ref{lemma:upper bound of one step}, we obtain
\begin{equation}\label{equ:t3 not in}
\begin{split}
&\mathbb{E}\left\{\left\|\bar{\nu}_{k+1}^b-\theta^*\right\|_2^2\right\}
\\&\le(1-\mu\eta_{k})\mathbb{E}\left\{\left\|\bar{\theta}_k-\theta^*\right\|_2^2\right\}+6\lambda\eta_{k}^2\Gamma\\
&\quad+8\eta_{k}^2(E-1)^2G^2+4E^2\eta_{k}^2G^2\frac{N-b}{N-1}\frac{1}{b}.
\end{split}
\end{equation}
From Corollary~\ref{cor:noise_lap} and~\ref{cor:noise_gau}, we note that the sensitivity depends on the learning rate. Therefore, we divide the variance of the noise item into two parts, the learning rate, $\tilde{\eta}_t^2$, and other constants related to specific mechanisms.
 Since we assume $\eta_{k}\le2\eta_{k+E}\le2\eta_{k+E-1}$, we have $\tilde{\eta}_t^2\le 4\eta_{k}^2,\quad\text{if }k+1\in\mathcal{C}_E.$ As a result, $\mathbb{E}\left\{\left\|\mathbf{w}_t^b\right\|_2^2\right\}\le \eta_{k}^2\omega_1,\quad\text{if }k+1\in\mathcal{C}_E.$ 
 $\omega_1=C_\mathcal{M}E^2T_g^z $ and $C_\mathcal{M}$ represents a constant number related with the DP mechanism. For example, by substituting Corollary~\ref{cor:noise_lap}, we obtain $C_\mathcal{M}=8pb\frac{\xi_1^2}{N^2\epsilon^2}$ for Laplace mechanism; by substituting Corollary~\ref{cor:noise_gau}, we obtain $C_\mathcal{M}=8pc_2^2\log(1/\delta)\frac{\xi_2^2}{N\epsilon^2}$ for Gaussian mechanism.
Therefore, plugging (\ref{equ:t3 if in}) and (\ref{equ:t3 not in}) into (\ref{equ:t3 update}), we obtain
\begin{equation*}
Y_{k+1}\le
\begin{cases}
(1-\mu\eta_{k})Y_k+\eta_{k}^2\omega_0,&\text{if }k+1\notin\mathcal{C}_E,\\
(1-\mu\eta_{k})Y_k+\eta_{k}^2\omega_0+\eta_{k}^2\omega_1,&\text{if }k+1\in\mathcal{C}_E,\\
\end{cases}
\end{equation*}
where
\begin{equation*}
\begin{split}
&\omega_0=6\lambda\Gamma+8(E-1)^2G^2+4E^2G^2\frac{N-b}{N-1}\frac{1}{b}.\\
\end{split}
\end{equation*}
Note here we slightly enlarge the upper bound when $k+1\notin\mathcal{C}_E$ to highlight the impact of noise item, because $4E^2G^2\frac{N-b}{N-1}\frac{1}{b}$ does not cause any significant impact on the final convergence rate. 
Let $\eta_{k}=\frac{\alpha}{k+\gamma}$, where $\mu\alpha\ge2$ and $\gamma>1$ such that $\eta_k\le\min(\frac{1}{\mu},\frac{1}{4\lambda})$ and $\eta_{k}\le2\eta_{k+E}$, we claim that
\begin{equation}\label{equ:t3 res}
Y_{k}\le\frac{\varphi}{k+\gamma}+\frac{\alpha^2t}{(k+\gamma-1)^2}\omega_1,
\end{equation}
where $\varphi=\max(\frac{\alpha^2}{\alpha\mu-1}\omega_0,\gamma Y_0)$, $t=\lfloor\frac{k}{E}\rfloor$.
We will use induction to prove this. First of all, when $k=0$, this is obviously true according to the definition of $\varphi$. Assume this holds for $Y_k$. If $k+1\notin\mathcal{C}_E$, we have
\begin{equation}\label{equ:t3 not in step1}
\begin{split}
Y_{k+1}
&\le(1-\mu\eta_{k})Y_k+\eta_{k}^2\omega_0\\
&\le(\frac{k+\gamma-\mu\alpha}{k+\gamma})\left(\frac{\varphi}{k+\gamma}+\frac{\alpha^2t}{(k+\gamma-1)^2}\omega_1\right)\\
&\quad+\left(\frac{\alpha}{k+\gamma}\right)^2\omega_0\\
&=\frac{k+\gamma-1}{(k+\gamma)^2}\varphi+\frac{\alpha^2}{(k+\gamma)^2}\omega_0-\frac{\mu\alpha-1}{(k+\gamma)^2}\varphi\\
&\quad+\frac{(k+\gamma-\mu\alpha)\alpha^2t}{(k+\gamma)(k+\gamma-1)^2}\omega_1.\\
\end{split}
\end{equation}
Using that $\varphi\ge\frac{\alpha^2}{\mu\alpha-1}$, we obtain
\begin{equation}\label{equ:t3 not in step2}
\begin{split}
&\frac{k+\gamma-1}{(k+\gamma)^2}\varphi+\frac{\alpha^2}{(k+\gamma)^2}\omega_0-\frac{\mu\alpha-1}{(k+\gamma)^2}\varphi\\
&\le \frac{k+\gamma-1}{(k+\gamma)^2}\varphi \le \frac{\varphi}{k+\gamma+1}.
\end{split}
\end{equation}
Recalling that $\mu\alpha\ge2$, for the last term of (\ref{equ:t3 not in step1}), we have
\begin{equation}\label{equ:t3 not in step3}
\begin{split}
\frac{(k+\gamma-\mu\alpha)\alpha^2t}{(k+\gamma)(k+\gamma-1)^2}\omega_1
&\le\frac{(k+\gamma-2)}{(k+\gamma-1)^2-1}\frac{\alpha^2t}{k+\gamma}\omega_1\\
&=\frac{\alpha^2t}{(k+\gamma)^2}\omega_1.\\
\end{split}
\end{equation}
Plugging (\ref{equ:t3 not in step2}) and (\ref{equ:t3 not in step3}) into (\ref{equ:t3 not in step1}), we finally obtain
\begin{equation*}
Y_{k+1}\le \frac{\varphi}{k+\gamma+1}+\frac{\alpha^2t}{(k+\gamma)^2}\omega_1.
\end{equation*}

If $k+1\in\mathcal{C}_E$, we have
\begin{equation*}
\begin{split}
Y_{k+1}
&\le(1-\mu\eta_{k})Y_k+\eta_{k}^2\omega_0\\
&\le\left(\frac{k+\gamma-\mu\alpha}{k+\gamma}\right)\left(\frac{\varphi}{k+\gamma}+\frac{\alpha^2t}{(k+\gamma-1)^2}\omega_1\right)\\
&\quad+\left(\frac{\alpha}{k+\gamma}\right)^2\omega_0+\left(\frac{\alpha}{k+\gamma}\right)^2\omega_1\\
&=\frac{k+\gamma-1}{(k+\gamma)^2}\varphi+\frac{\alpha^2}{(k+\gamma)^2}\omega_0-\frac{\mu\alpha-1}{(k+\gamma)^2}\varphi\\
&\quad+\frac{(k+\gamma-\mu\alpha)\alpha^2t}{(k+\gamma)(k+\gamma-1)^2}\omega_1+\left(\frac{\alpha}{k+\gamma}\right)^2\omega_1.\\
\end{split}
\end{equation*}
The first three terms are as same as in (\ref{equ:t3 not in step1}); so we can use (\ref{equ:t3 not in step2}) again.

Recalling that $t=\lfloor \frac{k}{E}\rfloor$ and $k+1\in\mathcal{C}_E$, which implies that $\lfloor \frac{k+1}{E}\rfloor=\lfloor \frac{k}{E}\rfloor+1=t+1$, we only need to show
\begin{equation*}
\begin{split}
\frac{(k+\gamma-\mu\alpha)\alpha^2t}{(k+\gamma)(k+\gamma-1)^2}\omega_1+\left(\frac{\alpha}{k+\gamma}\right)^2\omega_1&\le\frac{\alpha^2(t+1)}{(k+\gamma)^2}\omega_1\\
\frac{(k+\gamma-\mu\alpha)\alpha^2t}{(k+\gamma)(k+\gamma-1)^2}\omega_1&\le \frac{\alpha^2t}{(k+\gamma)^2}\omega_1,\\
\end{split}
\end{equation*}
which has been proved in (\ref{equ:t3 not in step3}). Therefore, we proved that (\ref{equ:t3 res}) holds both when $k+1\notin\mathcal{C}_E$ and $k+1\in\mathcal{C}_E$. 

Let $\alpha=\frac{2}{\mu}$ and $\gamma=\max(E,8\frac{\lambda}{\mu})$, we have $\eta_k=\frac{2}{\mu}\frac{1}{k+\gamma}$. Therefore, we finally obtain
\begin{equation*}
    Y_k\le\frac{1}{k+\gamma}\left(\frac{4}{\mu^2}\omega_0+\gamma Y_0\right)+\frac{4}{\mu^2}\frac{t}{(k+\gamma-1)^2}\omega_1.
\end{equation*}
\end{proof}

\ifCLASSOPTIONcaptionsoff
  \newpage
\fi



%


\bibliographystyle{IEEEtran}
\bibliography{IEEEbib}

\begin{thebibliography}{10}
\providecommand{\url}[1]{#1}
\csname url@samestyle\endcsname
\providecommand{\newblock}{\relax}
\providecommand{\bibinfo}[2]{#2}
\providecommand{\BIBentrySTDinterwordspacing}{\spaceskip=0pt\relax}
\providecommand{\BIBentryALTinterwordstretchfactor}{4}
\providecommand{\BIBentryALTinterwordspacing}{\spaceskip=\fontdimen2\font plus
\BIBentryALTinterwordstretchfactor\fontdimen3\font minus
  \fontdimen4\font\relax}
\providecommand{\BIBforeignlanguage}[2]{{%
\expandafter\ifx\csname l@#1\endcsname\relax
\typeout{** WARNING: IEEEtran.bst: No hyphenation pattern has been}%
\typeout{** loaded for the language `#1'. Using the pattern for}%
\typeout{** the default language instead.}%
\else
\language=\csname l@#1\endcsname
\fi
#2}}
\providecommand{\BIBdecl}{\relax}
\BIBdecl

\bibitem{yeom2018privacy}
S.~Yeom, I.~Giacomelli, M.~Fredrikson, and S.~Jha, ``Privacy risk in machine
  learning: Analyzing the connection to overfitting,'' in \emph{2018 IEEE 31st
  Computer Security Foundations Symposium (CSF)}.\hskip 1em plus 0.5em minus
  0.4em\relax IEEE, 2018, pp. 268--282.

\bibitem{hitaj2017deep}
B.~Hitaj, G.~Ateniese, and F.~Perez-Cruz, ``Deep models under the gan:
  information leakage from collaborative deep learning,'' in \emph{Proceedings
  of the 2017 ACM SIGSAC Conference on Computer and Communications Security},
  2017, pp. 603--618.

\bibitem{pmlr-v54-mcmahan17a}
\BIBentryALTinterwordspacing
B.~McMahan, E.~Moore, D.~Ramage, S.~Hampson, and B.~A. y~Arcas,
  ``{Communication-Efficient Learning of Deep Networks from Decentralized
  Data},'' in \emph{Proceedings of the 20th International Conference on
  Artificial Intelligence and Statistics}, ser. Proceedings of Machine Learning
  Research, A.~Singh and J.~Zhu, Eds., vol.~54.\hskip 1em plus 0.5em minus
  0.4em\relax Fort Lauderdale, FL, USA: PMLR, 20--22 Apr 2017, pp. 1273--1282.
  [Online]. Available: \url{http://proceedings.mlr.press/v54/mcmahan17a.html}
\BIBentrySTDinterwordspacing

\bibitem{wei2020framework}
W.~Wei, L.~Liu, M.~Loper, K.-H. Chow, M.~E. Gursoy, S.~Truex, and Y.~Wu, ``A
  framework for evaluating gradient leakage attacks in federated learning,''
  \emph{arXiv preprint arXiv:2004.10397}, 2020.

\bibitem{zhao2020idlg}
B.~Zhao, K.~R. Mopuri, and H.~Bilen, ``idlg: Improved deep leakage from
  gradients,'' \emph{arXiv preprint arXiv:2001.02610}, 2020.

\bibitem{zhu2019deep}
L.~Zhu, Z.~Liu, and S.~Han, ``Deep leakage from gradients,'' in \emph{Advances
  in Neural Information Processing Systems}, 2019, pp. 14\,747--14\,756.

\bibitem{236216}
\BIBentryALTinterwordspacing
N.~Carlini, C.~Liu, {\'U}.~Erlingsson, J.~Kos, and D.~Song, ``The secret
  sharer: Evaluating and testing unintended memorization in neural networks,''
  in \emph{28th {USENIX} Security Symposium ({USENIX} Security 19)}.\hskip 1em
  plus 0.5em minus 0.4em\relax Santa Clara, CA: {USENIX} Association, Aug.
  2019, pp. 267--284. [Online]. Available:
  \url{https://www.usenix.org/conference/usenixsecurity19/presentation/carlini}
\BIBentrySTDinterwordspacing

\bibitem{8737416}
Z.~{Wang}, M.~{Song}, Z.~{Zhang}, Y.~{Song}, Q.~{Wang}, and H.~{Qi}, ``Beyond
  inferring class representatives: User-level privacy leakage from federated
  learning,'' in \emph{IEEE INFOCOM 2019 - IEEE Conference on Computer
  Communications}, 2019, pp. 2512--2520.

\bibitem{bhowmick2018protection}
A.~Bhowmick, J.~C. Duchi, J.~Freudiger, G.~Kapoor, and R.~Rogers, ``Protection
  against reconstruction and its applications in private federated learning,''
  \emph{ArXiv}, vol. abs/1812.00984, 2018.

\bibitem{wu2019value}
\BIBentryALTinterwordspacing
N.~Wu, F.~Farokhi, D.~Smith, and M.~Kaafar, ``The value of collaboration in
  convex machine learning with differential privacy,'' in \emph{2020 IEEE
  Symposium on Security and Privacy (SP)}.\hskip 1em plus 0.5em minus
  0.4em\relax Los Alamitos, CA, USA: IEEE Computer Society, may 2020, pp.
  485--498. [Online]. Available:
  \url{https://doi.ieeecomputersociety.org/10.1109/SP40000.2020.00025}
\BIBentrySTDinterwordspacing

\bibitem{wei2019federated}
K.~{Wei}, J.~{Li}, M.~{Ding}, C.~{Ma}, H.~H. {Yang}, F.~{Farokhi}, S.~{Jin},
  T.~Q.~S. {Quek}, and H.~V. {Poor}, ``Federated learning with differential
  privacy: Algorithms and performance analysis,'' \emph{IEEE Transactions on
  Information Forensics and Security}, 2020.

\bibitem{kairouz2019advances}
P.~Kairouz, H.~B. McMahan, B.~Avent, A.~Bellet, M.~Bennis, A.~N. Bhagoji,
  K.~Bonawitz, Z.~Charles, G.~Cormode, R.~Cummings \emph{et~al.}, ``Advances
  and open problems in federated learning,'' \emph{arXiv preprint
  arXiv:1912.04977}, 2019.

\bibitem{geyer2017differentially}
R.~C. Geyer, T.~Klein, and M.~Nabi, ``Differentially private federated
  learning: A client level perspective,'' \emph{arXiv preprint
  arXiv:1712.07557}, 2017.

\bibitem{7113353}
Z.~{Jorgensen}, T.~{Yu}, and G.~{Cormode}, ``Conservative or liberal?
  personalized differential privacy,'' in \emph{2015 IEEE 31st International
  Conference on Data Engineering}, 2015, pp. 1023--1034.

\bibitem{10.1145/2676726.2677005}
\BIBentryALTinterwordspacing
H.~Ebadi, D.~Sands, and G.~Schneider, ``Differential privacy: Now it's getting
  personal,'' in \emph{Proceedings of the 42nd Annual ACM SIGPLAN-SIGACT
  Symposium on Principles of Programming Languages}, ser. POPL '15.\hskip 1em
  plus 0.5em minus 0.4em\relax New York, NY, USA: Association for Computing
  Machinery, 2015, p. 69–81. [Online]. Available:
  \url{https://doi.org/10.1145/2676726.2677005}
\BIBentrySTDinterwordspacing

\bibitem{googleFL}
\BIBentryALTinterwordspacing
B.~McMahan and D.~Ramage, ``Federated learning: Collaborative machine learning
  without centralized training data,'' Google AI Blog, 2017. [Online].
  Available:
  \url{https://ai.googleblog.com/2017/04/federated-learning-collaborative.html}
\BIBentrySTDinterwordspacing

\bibitem{bonawitz2019towards}
K.~Bonawitz, H.~Eichner, W.~Grieskamp, D.~Huba, A.~Ingerman, V.~Ivanov,
  C.~Kiddon, J.~Kone\v{c}n\'{y}, S.~Mazzocchi, B.~McMahan, T.~Van~Overveldt,
  D.~Petrou, D.~Ramage, and J.~Roselander, ``Towards federated learning at
  scale: System design,'' in \emph{Proceedings of Machine Learning and
  Systems}, A.~Talwalkar, V.~Smith, and M.~Zaharia, Eds., pp. 374--388.

\bibitem{Li2020On}
\BIBentryALTinterwordspacing
X.~Li, K.~Huang, W.~Yang, S.~Wang, and Z.~Zhang, ``On the convergence of fedavg
  on non-iid data,'' in \emph{International Conference on Learning
  Representations}, 2020. [Online]. Available:
  \url{https://openreview.net/forum?id=HJxNAnVtDS}
\BIBentrySTDinterwordspacing

\bibitem{Li2020FederatedLC}
T.~Li, A.~K. Sahu, A.~Talwalkar, and V.~Smith, ``Federated learning:
  Challenges, methods, and future directions,'' \emph{IEEE Signal Processing
  Magazine}, vol.~37, pp. 50--60, 2020.

\bibitem{10.1007/11787006_1}
C.~Dwork, ``Differential privacy,'' in \emph{Automata, Languages and
  Programming}, M.~Bugliesi, B.~Preneel, V.~Sassone, and I.~Wegener, Eds.\hskip
  1em plus 0.5em minus 0.4em\relax Berlin, Heidelberg: Springer Berlin
  Heidelberg, 2006, pp. 1--12.

\bibitem{8290673}
H.~{Shin}, S.~{Kim}, J.~{Shin}, and X.~{Xiao}, ``Privacy enhanced matrix
  factorization for recommendation with local differential privacy,''
  \emph{IEEE Transactions on Knowledge and Data Engineering}, vol.~30, no.~9,
  pp. 1770--1782, 2018.

\bibitem{ding2017collecting}
\BIBentryALTinterwordspacing
B.~Ding, J.~Kulkarni, and S.~Yekhanin, ``Collecting telemetry data privately,''
  in \emph{Advances in Neural Information Processing Systems 30}, December
  2017. [Online]. Available:
  \url{https://www.microsoft.com/en-us/research/publication/collecting-telemetry-data-privately/}
\BIBentrySTDinterwordspacing

\bibitem{li2020differentially}
\BIBentryALTinterwordspacing
J.~Li, M.~Khodak, S.~Caldas, and A.~Talwalkar, ``Differentially private
  meta-learning,'' in \emph{International Conference on Learning
  Representations}, 2020. [Online]. Available:
  \url{https://openreview.net/forum?id=rJgqMRVYvr}
\BIBentrySTDinterwordspacing

\bibitem{NEURIPS2019_fc0de4e0}
E.~Bagdasaryan, O.~Poursaeed, and V.~Shmatikov, ``Differential privacy has
  disparate impact on model accuracy,'' in \emph{Advances in Neural Information
  Processing Systems}, H.~Wallach, H.~Larochelle, A.~Beygelzimer,
  F.~d\textquotesingle Alch\'{e}-Buc, E.~Fox, and R.~Garnett, Eds.\hskip 1em
  plus 0.5em minus 0.4em\relax Curran Associates, Inc., pp. 15\,479--15\,488.

\bibitem{brendan2018learning}
\BIBentryALTinterwordspacing
H.~B. McMahan, D.~Ramage, K.~Talwar, and L.~Zhang, ``Learning differentially
  private recurrent language models,'' in \emph{International Conference on
  Learning Representations}, 2018. [Online]. Available:
  \url{https://openreview.net/forum?id=BJ0hF1Z0b}
\BIBentrySTDinterwordspacing

\bibitem{pihur2018differentiallyprivate}
V.~Pihur, A.~Korolova, F.~Liu, S.~Sankuratripati, M.~Yung, D.~Huang, and
  R.~Zeng, ``Differentially-private "draw and discard" machine learning,''
  \emph{ArXiv}, vol. abs/1807.04369, 2018.

\bibitem{10.1145/3378679.3394533}
\BIBentryALTinterwordspacing
S.~Truex, L.~Liu, K.-H. Chow, M.~E. Gursoy, and W.~Wei, ``Ldp-fed: Federated
  learning with local differential privacy,'' in \emph{Proceedings of the Third
  ACM International Workshop on Edge Systems, Analytics and Networking}, ser.
  EdgeSys ’20.\hskip 1em plus 0.5em minus 0.4em\relax New York, NY, USA:
  Association for Computing Machinery, 2020, p. 61–66. [Online]. Available:
  \url{https://doi.org/10.1145/3378679.3394533}
\BIBentrySTDinterwordspacing

\bibitem{seif2020wireless}
M.~{Seif}, R.~{Tandon}, and M.~{Li}, ``Wireless federated learning with local
  differential privacy,'' in \emph{2020 IEEE International Symposium on
  Information Theory (ISIT)}, 2020, pp. 2604--2609.

\bibitem{Abadi_2016}
\BIBentryALTinterwordspacing
M.~Abadi, A.~Chu, I.~Goodfellow, H.~B. McMahan, I.~Mironov, K.~Talwar, and
  L.~Zhang, ``Deep learning with differential privacy,'' \emph{Proceedings of
  the 2016 ACM SIGSAC Conference on Computer and Communications Security}, Oct
  2016. [Online]. Available: \url{http://dx.doi.org/10.1145/2976749.2978318}
\BIBentrySTDinterwordspacing

\bibitem{liu2020fedsel}
R.~Liu, Y.~Cao, M.~Yoshikawa, and H.~Chen, ``Fedsel: Federated sgd under local
  differential privacy with top-k dimension selection,'' in \emph{Database
  Systems for Advanced Applications}, Y.~Nah, B.~Cui, S.-W. Lee, J.~X. Yu,
  Y.-S. Moon, and S.~E. Whang, Eds.\hskip 1em plus 0.5em minus 0.4em\relax
  Cham: Springer International Publishing, 2020, pp. 485--501.

\bibitem{dwork2014algorithmic}
C.~Dwork, A.~Roth \emph{et~al.}, ``The algorithmic foundations of differential
  privacy,'' \emph{Foundations and Trends{\textregistered} in Theoretical
  Computer Science}, vol.~9, no. 3--4, pp. 211--407, 2014.

\bibitem{JMLR:v19:17-650}
\BIBentryALTinterwordspacing
R.~Leblond, F.~Pedregosa, and S.~Lacoste-Julien, ``Improved asynchronous
  parallel optimization analysis for stochastic incremental methods,''
  \emph{Journal of Machine Learning Research}, vol.~19, no.~81, pp. 1--68,
  2018. [Online]. Available: \url{http://jmlr.org/papers/v19/17-650.html}
\BIBentrySTDinterwordspacing

\bibitem{pmlr-v80-nguyen18c}
\BIBentryALTinterwordspacing
L.~Nguyen, P.~H. NGUYEN, M.~van Dijk, P.~Richtarik, K.~Scheinberg, and
  M.~Takac, ``{SGD} and hogwild! {C}onvergence without the bounded gradients
  assumption,'' in \emph{Proceedings of the 35th International Conference on
  Machine Learning}, ser. Proceedings of Machine Learning Research, J.~Dy and
  A.~Krause, Eds., vol.~80.\hskip 1em plus 0.5em minus 0.4em\relax
  Stockholmsmässan, Stockholm Sweden: PMLR, 10--15 Jul 2018, pp. 3750--3758.
  [Online]. Available: \url{http://proceedings.mlr.press/v80/nguyen18c.html}
\BIBentrySTDinterwordspacing

\bibitem{Nguyen2019TightDI}
P.~NGUYEN, L.~Nguyen, and M.~van Dijk, ``Tight dimension independent lower
  bound on the expected convergence rate for diminishing step sizes in sgd,''
  in \emph{Advances in Neural Information Processing Systems}, H.~Wallach,
  H.~Larochelle, A.~Beygelzimer, F.~d\textquotesingle Alch\'{e}-Buc, E.~Fox,
  and R.~Garnett, Eds.\hskip 1em plus 0.5em minus 0.4em\relax Curran
  Associates, Inc., pp. 3665--3674.

\bibitem{caldas2018leaf}
S.~Caldas, P.~Wu, T.~Li, J.~Kone{\v{c}}n{\`y}, H.~B. McMahan, V.~Smith, and
  A.~Talwalkar, ``Leaf: A benchmark for federated settings,'' \emph{arXiv
  preprint arXiv:1812.01097}, 2018.

\end{thebibliography}

\end{document}